\documentclass[12pt]{conm-p-l}
\usepackage{amsmath,amssymb,latexsym,graphicx,citesort}
\makeindex

\newtheorem{theorem}{Theorem}[section]
\newtheorem{lemma}[theorem]{Lemma}
\newtheorem{corollary}[theorem]{Corollary}

\newtheorem{proposition}[theorem]{Proposition}

\newtheorem{definition}[theorem]{Definition}

\newtheorem{remark}[theorem]{Remark}

\newcommand{\RR}{\mathbb{R}}
\newcommand{\R}{\mathbb{R}}

\newcommand{\ZZ}{\mathbb{Z}}
\newcommand{\CC}{\mathbb{C}}
\newcommand{\C}{\mathbb{C}}

\newcommand{\TT}{\mathbb{T}}
\newcommand{\T}{\mathbb{T}}

\newcommand{\B}{\mathbf{B}}
\newcommand{\W}{\mathbf{W}}
\newcommand{\G}{\Gamma}
\newcommand{\g}{\gamma}

\newcommand{\fh}{\hat{f}}

\newcommand{\D}{\mathcal{D}}

\newcommand{\dn}{d\,'}

\newcommand{\bd}{\begin{definition}}
\newcommand{\ed}{\end{definition}}
\newcommand{\bt}{\begin{theorem}}
\newcommand{\et}{\end{theorem}}
\newcommand{\bl}{\begin{lemma}}
\newcommand{\el}{\end{lemma}}
\newcommand{\bc}{\begin{corollary}}
\newcommand{\ec}{\end{corollary}}
\newcommand{\br}{\begin{remark}}
\newcommand{\er}{\end{remark}}
\newcommand{\bp}{\begin{proposition}}
\newcommand{\ep}{\end{proposition}}

\begin{document}
\title[Wannier functions frames]{On Parseval frames of exponentially decaying composite Wannier functions}
\author{David Auckly}
\address{Mathematics Department,
Kansas State University,
Manhattan, KS, USA}
\email{dav@math.ksu.edu}
\author{Peter Kuchment}
\address{Mathematics Department,
Texas A\&M University,
College Station, TX 77843-3368, USA}
\email{kuchment@math.tamu.edu}
\thanks{The second author was supported in part by NSF grants DMS-0406022 and DMS-1517938. }

\subjclass{Primary 35Q40, 35P10, 47F05, 81V55; Secondary 65N25}
\date{}

\keywords{Spectral theory, periodic operator, condensed matter, Bloch bundle, Wannier function}

\begin{abstract}

Let $L$ be a periodic self-adjoint linear elliptic operator in $\R^n$  with coefficients periodic with respect to a lattice $\G$, e.g. Schr\"{o}dinger operator $(i^{-1}\partial/\partial_x-A(x))^2+V(x)$ with periodic magnetic and electric potentials $A,V$, or a Maxwell operator $\nabla\times\varepsilon (x)^{-1}\nabla\times$ in a periodic medium. Let also $S$ be a finite part of its spectrum separated by gaps from the rest of the spectrum. We address here the question of existence of a finite set of exponentially decaying Wannier functions $w_j(x)$ such that their $\G$-shifts $w_{j,\g}(x)=w_j(x-\g)$ for $\g\in\G$ span the whole spectral subspace corresponding to $S$. It was shown by D.~Thouless in 1984 that a topological obstruction sometimes exists to finding exponentially decaying $w_{j,\g}$ that form an orthonormal (or any) basis of the spectral subspace. This obstruction has the form of non-triviality of certain finite dimensional (with the dimension equal to the number of spectral bands in $S$) analytic vector bundle (Bloch bundle), which we denote $\Lambda_S$. It was shown by G.~Nenciu in 1983 that in the presence of time reversal symmetry (which implies absence of magnetic fields), and if $S$ is a single band, the bundle is trivial and thus the desired Wannier functions do exist. In 2007, G.~Panati proved that in dimensions $n\leq 3$, even if $S$ consists of several spectral bands, the time reversal symmetry removes the obstruction as well. If the bundle is non-trivial, it was shown in 2009 by one of the authors that it is always possible to find a finite number $l$ (estimated there as $m\leq l \leq 2^n m$) of exponentially decaying Wannier functions $w_j$ such that their $\G$-shifts form a tight (Parseval) frame in the spectral subspace. A Parseval frame is the next best thing after an orthonormal basis (unavailable in the presence of the topological obstacle). This appears to be the best one can do when the topological obstruction is present. Here we significantly improve the estimate on the number of extra Wannier functions needed, showing that in physical dimensions the number $l$ can be chosen equal to $m+1$, i.e. only one extra family of Wannier functions is required. This is the lowest number possible in the presence of the topological obstacle. The result for dimension four is also stated (without a proof), in which case $m+2$ functions are needed.

One should mention that a powerful numerical machinery for creating bases and frames of decaying Wannier functions has been developed in \cite{Marzari,Marzari12,MarzSou,Brouder}, and studied analytically in \cite{Panati,Panati13,Monaco_etal,Fior,Busch_wannier3}. Recent progress in direction of actual construction of Wannier bases and frames has also been achieved \cite{Cances,Cornean,CornMonaco,Fior}. We do not consider his issue, except showing in Section \ref{S:generic} that a ``generic'' choice of the missing $(m+1)$st function should work.

The main result of the paper was announced without proof in \cite[Section 6.5]{KuchBAMS}.
\end{abstract}

\maketitle

\section{Introduction}
Wannier functions, along with Bloch waves, play an important role in various areas of physics and material science, condensed matter theory, photonic crystal theory (see, e.g., \cite{JJWM,Kuch_photchapter} for general discussion of photonic crystals). They give a very useful tool for description of electronic properties of solids, theory of polarization, photonic crystals, numerical analysis using tight-binding approximation, etc. (see, e.g. \cite{Wann,Whit,Bush_wannier,Busch_wannier2,KuchBAMS,Kuch_wan,Busch_wannier3,AM,Nenciu,Brouder,Panati,Marzari,Kohn,KohnLutt,MarzSou,Wann,Wann_web} and references therein for Wannier functions and their applications). As it is formulated in \cite{MarzSou}, strongly localized Wannier functions ``are the solid-state equivalent of ``localized molecular orbitals''..., and thus provide an insightful picture of the nature of chemical bonding.'' It is crucial to have the Wannier functions decaying as fast as possible. Thus the problem of choosing a finite number of exponentially decaying Wannier functions whose lattice shifts form an orthonormal basis in the spectral subspace corresponding an isolated part $S$ of the spectrum has been intensively considered in physics literature since the paper by W.~Kohn \cite{Kohn}, who showed that this was possible in $1D$. The problem becomes non-trivial in $2D$ and higher dimensions. Indeed, a topological obstruction,  not present in $1D$, might arise, as shown by D.~Thouless \cite{Thouless}. Existence of such a basis is known to be equivalent to triviality of certain analytic vector fiber bundle (\textbf{Bloch bundle}), which we will denote $\Lambda_S$ (see (\ref{E:bundle})). G.~Nenciu \cite{Nenciu} showed in 1983 (see also \cite{Helffer}) presence of time reversal symmetry, and if $S$ is a single spectral band, the bundle is trivial and thus the obstacle disappears. G.~Panati \cite{Panati} proved in 2007, that in dimensions $n\leq 3$, even if $S$ consists of several spectral bands, the time reversal symmetry still removes the obstruction. In this case one resorts to the so called \textbf{composite}, or \textbf{generalized Wannier functions} that correspond to a finite family of bands, rather than to a single band (see Section \ref{S:prelim} below). The activity in this direction is still high and even increasing \cite{MarzSou,Wann_web}.

What can one do if the topological obstacle is present and thus a family of Wannier functions with the described basis properties cannot exist? One can try to relax the exponential decay condition, but there is not much one can do without stumbling upon the topological obstruction (e.g., a slow decay such as summability of $L^2$ norms over the shifted copies of the Wigner-Seitz cell is already impossible \cite{Kuch_wan}).

In the positive direction, the following result was proven in \cite{Kuch_wan}:

{\bf Theorem } {\em Let $L$ be a self-adjoint elliptic $\G$-periodic operator in $\RR^n, n\geq 1$ and $S\subset\RR$ be the union of $m$ spectral bands of $L$. Suppose that $S$ is separated from the rest of the spectrum by gaps. Then there exists a finite number $l$ of exponentially decaying composite Wannier functions
$w_j(x)$ such that their shifts $w_{j,\g}:=w_j(x-\g),\g\in\G$ form a tight (Parseval) frame in the spectral subspace $H_S$ of the operator $L$. This means that for any $f(x)\in H_S$, the equality holds
\begin{equation}
    \int\limits_{\RR^n}|f(x)|^2dx=\sum\limits_{j,\g}|\int\limits_{\RR^n}f(x)\overline{w_{j,\g}(x)}dx|^2.
\end{equation}
Here the number $l\in [m,2^n m]$ is equal to the smallest dimension of a trivial bundle containing an equivalent copy of $\Lambda_S$. In particular, $l=m$ if and only if $\Lambda_S$ is trivial, in which case an orthonormal basis of exponentially decaying composite Wannier functions exists.
}

Here a \textbf{frame} means an \underline{overdetermined} (rather than basis) system of functions, and thus the orthonormal property is not achievable. The tight (Parseval) property is the best analog of orthonormality one can get in this case (see, e.g., \cite{Larson_frame}). For instance, it allows the control of the $L^2$ norms in terms of the projections onto the Wannier system. Moreover, the overdeterminancy makes the computation stabler.

The significant deficiencies of the above result of \cite{Kuch_wan} are that, first, the upped bound $l\leq 2^nm$ is ridiculously high; second, even for a manageable value of $l>m$ it is not clear how to practically create the overdetermined system of $l$ Wannier functions, since the proof in \cite{Kuch_wan} is not constructive.
Even in the case of a trivial Bloch bundle, it is not that clear how one can actually find the needed holomorphic basis of Bloch functions, and this is still under investigation \cite{Cornean,Fior,CornMonaco,Read}.

We address the first of these issues in this text. In Section \ref{S:prelim} we survey the main notions and results of the theory that will be needed for the rest of the text. In order to make the text independent, this section is rather long. An expert can skip through it or go immediately to Section \ref{S:3D}, where it is shown that in the ``physical dimension'' $n=3$ only one extra Wannier functions is needed, i.e. $l=m+1$ (announced in \cite[Section 6.5]{KuchBAMS}).  Section \ref{S:d4} contains (without a proof) an analog of this result, when $n=4$. Here $m+2$ Wannier functions might be required. We discuss how one can try to find the ``missing'' $(m+1)$st function in Section \ref{S:generic}. In Section \ref{S:remarks} final remarks and conclusions are provided. The texts ends with the Acknowledgments section.

\section{Main notions and auxiliary results}\label{S:prelim}

Let $L(x,D)$ be a bounded from below self-adjoint elliptic operator in $\RR^n$. The specific nature of the operator will be irrelevant (e.g., matrix operators, such as Dirac or Maxwell can be allowed). One can think, without loss of generality, of one's favorite periodic operator, e.g. the Schr\"{o}dinger operator $(\frac{1}{i}\frac{\partial}{\partial x}-A(x))^2+V(x)$ with real periodic magnetic and electric potentials $A,V$. Some conditions need to be imposed on the potentials to define a self-adjoint operator $L$ in $L^2(\RR^n)$ (e.g., \cite{Cycon,ReedSimon}). However, for what follows most of these details are not germane to the main issue, so we assume ``sufficiently nice'' (e.g., smooth) coefficients of $L$.

Let $\G$ be a (\textbf{Bravais} \cite{AM}) \textbf{lattice} in $\RR^n$, i.e. the set of integer linear combinations of vectors of a basis $a_1,\dots,a_n$ in $\RR^n$ \footnote{No generality will be lost if the reader assumes that $\G$ is the integer lattice $\ZZ^n$.}. The coefficients of $L$ are assumed to be periodic with respect to the shifts by vectors $\g\in\G$. We fix a \textbf{fundamental domain} $\W$ of $\G$, i.e. such that its $\G$-shifts cover the whole $\RR^n$ with only boundary overlap.

Let $\G^*$ the \textbf{reciprocal} (or \textbf{dual}) lattice \cite{AM} to $\G$. It lives in the dual space $(\RR^n)^*$, but if an inner product $(\cdot,\cdot)$ in $\RR^n$ is fixed, $\G^*$ can be realized in $\RR^n$ as the set of all vectors $\kappa$ such that $(\kappa,\g)\in 2\pi\ZZ$ for all $\g\in\G$ (if $\Gamma=\ZZ^n$, then $\Gamma^*=2\pi\ZZ^n$). We fix a fundamental domain $\B$ for $\G$ (e.g., the \textbf{first Brillouin zone} \cite{AM}). Then $\T:=\RR^n/\G$ ($\T^*:=\RR^n/\G^*$) is a torus and $\G$-($\G^*$-) periodic functions on $\RR^n$ are naturally identified with functions on $\T$ ($\T^*$).

The periodicity of the spectral problem $Lu=\lambda u$ with respect to $\G$ suggests to use the well known \textbf{Bloch-Floquet transform} (the name varies from a source to a source) \cite{Kuch_book,Kuch_photchapter,ReedSimon,Kuch_UMN,AM} :
\begin{equation}\label{E:Floquet_tr}
  f(x)\mapsto\hat{f}(k,x):=
  \sum\limits_{\g \in \G} f(x+\g) e^{-ik \cdot \g}.
\end{equation}
Here $k$ is a real (or complex) $n$-dimensional vector, which is called \textbf{quasi-momentum} (or \textbf{crystal momentum} and \textbf{Bloch momentum}). Assuming that $f$ decays sufficiently fast, there is no convergence problem. It is easy to check that for any (even complex) quasi-momentum $k$, the function $\hat{f}(k,x)$ is $\G^*$-periodic with respect to $k$ and is of the \textbf{Bloch} (also called {\bf Floquet}) \textbf{form} with respect to $x$, i.e.
\begin{equation}\label{E:Bloch}
\hat{f}(k,x)=e^{ik\cdot x}v_k(x),
\end{equation}
where $v_k(x)$ is $\G$-periodic. The values $x\in\W$ and $k\in \B$ are sufficient for determining the whole function $\hat{f}(k,x)$. One can consider $\hat{f}(k,x)$ as a function $\hat{f}(k,\cdot)$ on $\B$ (or better, on the torus $\T^*$) with values in a space of functions on $\W$. Considering the torus $\T^*$, it is more natural to consider $\hat{f}(k,\cdot)$ as a function of the {\bf Floquet multiplier} $z:=e^{ik}:=(e^{ik\cdot a_1},\dots,e^{ik\cdot a_n})$, rather than of the quasi-momentum $k$. Then the torus $\T^*$ becomes the unit torus $\{z| |z_j|=1,j=1,\dots,n\}$ in $\C^n$. We will also need some complex neighborhoods of the space of real quasi-momenta and of the torus $\TT^*$, defined for a given $\alpha >0$:
\begin{equation}\label{E:domain_D}
\D_\alpha = \{ k \in \CC^n|\, |\mbox{Im } k \cdot a_j| < \alpha, j=1, ...,
n\},
\end{equation}
and its image under the transform $k\mapsto z$
\begin{equation}\label{E:omega-a}
\Omega_\alpha=\{z=\left( z_1,...,z_n\right) \in
    \Omega \, | \, e^{-\alpha}<|z_j|<e^\alpha,\, j=1,...,n\}.
\end{equation}
Here $\{a_j\}$ is the basis of $\G$ mentioned before.

\begin{definition}\label{D:exp_space}\indent
\begin{itemize}
\item If $\Omega$ is an $n$-dimensional complex domain and $H$ is a Hilbert space, We will denote by $A(\Omega, H)$ the space of all $H$-valued analytic functions on $\Omega$, equipped with the topology of uniform convergence on compacta.
\item The space $L^2_\alpha(\RR^n)$ consists of all functions $f \in
L^2_{loc}(\RR^n)$ such that for any $0<b<\alpha$ the following
expression is finite:
\begin{equation}\label{E:L2_a}
  \mathop{sup}\limits_{\g\in\G}\, \|f \|_{L^2(\mathbf{W}+\g)}e^{b|\g|}<\infty.
\end{equation}
This space is equipped with the natural topology defined by
the semi-norms $\psi_b$.
\end{itemize}
\end{definition}

We now quote some standard results about the Bloch-Floquet transform (analogs of the standard Plancherel, Paley-Wiener, and inversion theorems for the Fourier series), see e.g., \cite{Kuch_book,KuchBAMS}:

\begin{theorem}\label{T:Planch}\indent
\begin{enumerate}
\item If $f \in L^2(\RR^n)$ and $K \subset \RR^n$ is a compact, then the
series (\ref{E:Floquet_tr}) converges in the space
$L^2(\TT^*,L^2(K))$. Moreover, the following equality (Plancherel
theorem) holds:
\begin{equation}\label{E:Planch}
  \|f\|_{_{L^2(\RR^n)}}^2=\int\limits_{\B}
  \|\fh (k,\cdot)\|_{_{L^2(\W)}}^2\dn k=\int\limits_{\TT^*}
  \|\fh (z,\cdot)\|_{_{L^2(\W)}}^2\dn z,
\end{equation}
where $\dn k$ is the normalized to total measure $1$ Lebesque measure on $\B$, and $\dn z$ is the normalized
Haar measure on $\TT^*$.
\item For any $\alpha \in (0,\infty]$, Bloch-Floquet transform
$$
f \mapsto \fh
$$
is a topological isomorphism of the space $L^2_\alpha(\RR^n)$ onto
$A(\Omega_\alpha,L^2(\mathbf{W}))$.
\item
For any $f \in L^2(\RR^n)$ the following inversion
formula holds:
\begin{equation}\label{E:Gelf_inversion}
f(x)=\int\limits_{\TT^*} \fh (k,x) \dn k, \,\, x \in \RR^n.
\end{equation}
\end{enumerate}
\end{theorem}

The first statement claims that the Bloch-Floquet transform is an isometry between the natural Hilbert spaces, the second shows that exponential decay transforms into analyticity in a neighborhood of the torus $\TT^*$ (a Paley-Wiener type theorem), and the third one provides an inversion of the transform.

\begin{remark}\label{R:aniso}
In the definition (\ref{E:omega-a}) of the domain $\Omega_\alpha$ one can allow a vector $\alpha=(\alpha_1,\dots,\alpha_n)$, modifying accordingly  (\ref{E:omega-a}) to

\begin{equation}\label{E:omega-avect}
\Omega_\alpha=\{z=\left( z_1,...,z_n\right) \in
    \Omega \, | \, e^{-\alpha_j}<|z_j|<e^\alpha_j,\, j=1,...,n\}.
\end{equation}
Then the second statement of Theorem \ref{T:Planch} still holds, if one modifies accordingly the definition (\ref{E:L2_a}) of the space $L^2_\alpha$, requiring anisotropic decay.
\end{remark}

The Bloch-Floquet transform block diagonalizes the periodic operator $L$, from which one obtains the well known \cite{ReedSimon,AM,Kuch_book,Kuch_UMN,KuchBAMS} \textbf{band-gap spectral structure} of the operator $L$. Namely,

\begin{theorem}\cite{Kuch_book,Wilcox}\label{T:bandgap}

Let $L(k)$ be the operator $L$ acting on the Bloch functions (\ref{E:Bloch}) with a fixed quasi-momentum $k$. If we label (for real $k$) the eigenvalues of $L$ in nondecreasing order as
\begin{equation}\label{E:fl_spectrum}
    \lambda_1(k)\leq \lambda_2(k)\leq \dots \mapsto \infty,
\end{equation}
then
\begin{enumerate}
  \item \textbf{band functions} $\lambda_j(k)$ are continuous, $\G^*$-periodic, and piece-wise analytic in $k$;
  \item if $I_j$ is the range of $\lambda_j(k)$ (the \textbf{$j$th band}), then the spectrum of $L$ is
  $$
  \sigma(L)=\bigcup\limits_j I_j.
  $$
  \item The corresponding \textbf{Bloch eigenfunctions} $\phi_j(k,x)$ (i.e., $L\phi_j=\lambda_j\phi_j$) can be chosen $\phi_j(k,\cdot)$ as a piece-wise analytic $L^2$-functions on $\B$ with values in $L^2(\W)$, whose norm in $L^2(\W)$ is almost everywhere constant and can be chosen equal to $1$. One can also assume that $\phi(k,\cdot)$ is $\G^*$-periodic with respect to $k$.
\end{enumerate}
\end{theorem}

\subsection{Bloch and Wannier functions}\label{SS:wannier}

There are two bases of functions (distributions) extremely useful for studying constant coefficients linear PDEs in $\R^n$. These are the basis of delta functions (well localized in the space) and the basis of plane waves (well localized in the dual space). They are dual to each other under the Fourier transform. One would want to have similar bases in the periodic case. Analogs of plane waves is clear - the Bloch functions. One suspects that an analog of delta functions could be obtained from these by the Bloch-Floquet transform. This indeed works (to some degree) and leads to the so called Wannier functions.

Let $\phi(k,x)$ be a Bloch function with the quasi-momentum $k$, depending ``sufficiently nicely'' on $k$. E.g., in our applications we can always assume that $\phi\in L^2(\TT^*,L^2(\W))$.

\begin{definition}
The {\em Wannier function} $w(x)$ corresponding to the Bloch function  $\phi(k,x)$ is
\begin{equation}\label{E:dwannier}
    w(x)=\int\limits_{\T^*} \phi(k,x) \dn k,\quad x\in\RR^n.
\end{equation}
\end{definition}

Comparing this definition with (\ref{E:Gelf_inversion}), one sees that the Wannier function $w(x)$ is just the inverse Bloch-Floquet transform of $\phi(k,x)$, and vice-versa, $\phi(k,x)$ is the Bloch-Floquet transform of $w(x)$.


It is often useful to have a Wannier function $w(x)$ normalized and having mutually orthogonal lattice shifts $w(x-\g)$. The following auxiliary result \cite{Kuch_wan} is useful:

\begin{lemma}\label{C:wannier_orthog}

\indent
\begin{enumerate}
 \item
The Wannier function $w(x)$ belongs to $L^2(\RR^n)$ and
\begin{equation}\label{E:wannier_l2}
    \int\limits_{\RR^n}|w(x)|^2dx=\int\limits_{\B} \|\phi(k,\cdot)\|^2_{L^2(\W)} \dn k.
\end{equation}
\item Functions $w_{\g}(x):=w(x-\g)$ are mutually orthogonal for $\g\in\G$ iff the function $\phi(k,x)$ in (\ref{E:dwannier}) has a $k$-independent norm in $L^2(\W)$.
\end{enumerate}
\end{lemma}


The most important property of Wannier functions, due to which one wants to use them for numerical calculations in many problems of physics is their decay. An exponential decay is often desired.

A direct consequence of Theorem \ref{T:Planch}) is that smoothness with respect to $k$ of the Bloch function $\phi_j(k,x)$ corresponds to decay of $w(x)$ \footnote{Smoothness also means matching the values and all derivatives across the boundaries of Brillouin zones.}:
\begin{lemma}\label{L:smooth_wannier}\cite[Section 2.2]{Kuch_book}
\indent
\begin{enumerate}
\item If $\sum\limits_{\g\in\G}\|w_j\|_{L^2(\W+\g)}<\infty$, then $\phi_j(k,\cdot )$ is a continuous $L^2(\W)$-valued function on $\T^*$.
\item Infinite differentiability of $\phi_j(k,\cdot )$ as a function on $\T^*$ with values in $L^2(\W)$ is equivalent to the decay of $\|w_j\|_{L^2(\W+\g)}$ faster than any power of $|\g|$ for $\g\in\G$.
\item Analyticity  of $\phi_j(k,\cdot )$ as a function on $\T^*$ with values in $L^2(\W)$ is equivalent to the exponential decay of $\|w_j\|_{L^2(\W+\g)}$.
\end{enumerate}
\end{lemma}

\subsection{Wannier bases for spectral subspaces}

Suppose that a compact set $S\subset\R$ is the union of $m$ ($m\geq 1)$ bands of the spectrum of $L$ (with overlaps allowed) and is separated by spectral gaps at the bottom and the top from the rest of the spectrum. An example is the part of the spectrum from its bottom till the first gap. We call such subset $S$ a {\bf composite band} (versus just a \textbf{band} when $m=1$).

Surrounding $S$ by a closed contour $\Sigma\in\C$, separating $S$ from the rest of the spectrum, one can define in the space $L_2(\R^n)$ the Riesz projector \cite{ReedSimon} $P$ onto the spectral subspace $H_S$ that corresponds to $S$. Restricting one's attention only to this subspace, one wants to have a convenient basis for numerical computations. Orthonormal bases of exponentially decaying Wannier functions are known to be the excellent candidates \cite{Marzari,Marzari12,MarzSou,Panati}. So, let us see how and whether one can find such bases.

Applying the above spectral projector to the direct integral decomposition of the operator $L$ (keeping the same contour $\Sigma$ for all $k\in B$), one gets and analytic with respect to $k$ family $P(k)$ of $m$-dimensional spectral projectors for the operators $L(k)$. It thus produces the $m$-dimensional \textbf{Bloch bundle} $\Lambda_S$ over the torus $\T^*$
\begin{equation}\label{E:bundle}
\Lambda_S=\bigcup_{\T^*}\emph{Ran}(P(k)),
\end{equation}
which corresponds to the spectral subset $S$. Then elements of the space $H_S$ correspond under Bloch-Floquet transform to $L_2$-sections of this bundle.

The discussion above (see more details in \cite{Kuch_wan}) shows that in order to have an orthonormal basis  of exponentially decaying Wannier functions in $H_S$, one need to find an \emph{analytic with respect to} $k$ basis $\phi_j(k,x)|_{j=1}^m$ of Bloch functions (sections of $\Lambda_S$) and take its (inverse) Bloch-Floquet transform. One immediately realises that (even if we were looking for a basis of continuous sections rather than analytic ones) this requires the bundle to be trivial. As it is shown in \cite{Kuch_wan}, the so called Oka's principle (implemented by Grauert \cite{Grauert1,Grauert2,Grauert3}) assures that there is no additional obstructions to getting an analytic, rather than just a continuous basis. Thus, topological non-triviality of this bundle is the only possible obstacle to the existence of desirable Wannier bases. It was shown by D.~Thouless \cite{Thouless}, that such obstacle can indeed occur. (For topological insulators this is not an obstacle, but rather a blessing \cite{Bern}.)

There are cases when the triviality is known. E.g., G.~Nenciu \cite{Nenciu} showed in 1983 (see also \cite{Helffer}) that when $m=1$ (single band), presence of time reversal symmetry guarantees triviality of the bundle and thus the obstacle disappears. G.~Panati \cite{Panati} proved in 2007, that in dimensions $n\leq 3$, even if $m>1$, the time reversal symmetry still removes the obstruction. However, what can one do if the obstruction exists? In \cite{Kuch_wan} one of the authors showed that if one allows some overdeterminancy of Wannier functions (rather than basis property), the obstacle disappears. In order to describe the result, we need to introduce some (rather popular nowadays) notions concerning frames.

\subsection{Parseval frames}\label{SS:parseval}

We recall here the notion of the so called {\bf tight (or Parseval) frame} of vectors (e.g., \cite{Larson_frame}) that replaces orthonormality in the overdetermined case.

\bd\label{D:frame}
\indent
\begin{itemize}
  \item A set of vectors $v_j,j=1,2,\dots$ in a Hilbert space $H$ is said to be a {\bf frame}, if for some constants $A,B>0$ and any vector $u\in H$ the following inequality holds:
      $$
      A\sum\limits_j |(u,v_j)|^2\leq \|u\|^2\leq B\sum\limits_j |(u,v_j)|^2.
      $$
  \item This set is said to be a {\bf tight (or Parseval) frame} if $A,B=1$, i.e. for any vector $u\in H$ the equality holds
  $$
  \sum\limits_j |(u,v_j)|^2 = \|u\|^2.
  $$
\end{itemize}
\ed

The name ``Parseval frame'' is justified by the fact that the formulas for expansion and synthesis for Parseval frames are exactly the same as for orthonormal bases. This follows from the simple observation \cite{Larson_frame} that tight frames are exactly the orthogonal projections of orthonormal bases from larger Hilbert spaces. At the same time, the overdetermined nature of the collection of vectors provides improved stability with respect to errors.

\subsection{Wannier Parseval frames}

The following result was proven in \cite{Kuch_wan}:
\begin{theorem}\label{T:kucwan} Let $L$ be a self-adjoint elliptic $\G$-periodic operator in $\RR^n, n\geq 1$ and $S\subset\RR$ be the union of $m$ spectral bands of $L$. Suppose that $S$ is separated from the rest of the spectrum. Then there exists a finite number $l$ of exponentially decaying composite Wannier functions
$w_j(x)$ such that their shifts $w_{j,\g}:=w_j(x-\g),\g\in\G$ form a tight (Parseval) frame in the spectral subspace $H_S$ of the operator $L$. This means that for any $f(x)\in H_S$, the equality holds
\begin{equation}
    \int\limits_{\RR^n}|f(x)|^2dx=\sum\limits_{j,\g}|\int\limits_{\RR^n}f(x)\overline{w_{j,\g}(x)}dx|^2.
\end{equation}
Here the number $l\in [m,2^n m]$ is equal to the smallest dimension of a trivial bundle containing an equivalent copy of $\Lambda_S$. In particular, $l=m$ if and only if $\Lambda_S$ is trivial, in which case an orthonormal basis of exponentially decaying composite Wannier functions exists.
\end{theorem}

Our goal here is to prove the result announced in \cite[section 6.5]{KuchBAMS}, which in physical dimensions drastically reduces the number of extra Wannier functions needed.

\section[3D]{In dimension three, only one extra Wannier function is needed}\label{S:3D}

As the techniques and results of \cite{Kuch_wan} surveyed in the section \ref{S:prelim} show,
the main issue is finding the smallest rank of a trivial bundle containing a copy of the Bloch bundle
in question. Moreover, due to the Grauert theorem \cite{Grauert1,Grauert2,Grauert3} (an instance of the Oka's principle), only topological triviality is needed\footnote{For basic notions involved in this and the next sections, one can refer to \cite{Milnor,Steenrod,Husemoller}.}.

When the dimension of the base does not exceed $3$, the answer is given by the following,
probably known, theorem, a proof of which we provide for completeness. All complexes below are assumed to be compact.

\begin{theorem}\label{T:3Dtriv}
Any complex vector bundle of rank $m$ over any $3$-dimensional CW complex $X$ embeds into a
trivial bundle of rank $m+1$.
\end{theorem}
Let us prove first the following auxiliary statement:

\begin{lemma}\label{L:1d} Any complex vector bundle of rank larger than $1$ over a $3$-dimensional CW complex $X$ admits a
non-vanishing section.
\end{lemma}
\begin{proof}
Let us introduce a Riemannian metric on the bundle. Clearly one can find a non-zero (in fact unit
length) section over the $0$-skeleton of $X$. Let us proceed by induction with respect to the dimension of the skeleton.
Assuming there is a section over the $k$-skeleton,
consider extending to the $(k + 1)$-skeleton. By induction hypothesis, we have a section over the
boundary of each (k + 1)-cell. Such may be viewed as a map from the boundary of the
cell to the unit sphere in the fiber, i.e. $S^k\mapsto S^{2m-1}$. However, for $m>1$, the fundamental group
and the second homotopy group of $S^{2m-1}$ are trivial, so the section can be extended.
\end{proof}
We can now return to proving Theorem \ref{T:3Dtriv}.

\begin{proof} Lemma \ref{L:1d} shows that any complex vector bundle $E$ of rank greater
than $1$ over a $3$-complex contains a trivial sub-bundle of rank $1$. I.e., $E$ is the direct
sum of a trivial bundle and its orthogonal complement. By induction, any
complex vector $m$-bundle over a $3$-complex is isomorphic to $L\bigoplus \C^{m-1}$,
where $L$ is a linear bundle.

Computing the first Chern class, we to obtain $c_1(L\bigoplus \C^{m-1})=c_1(L)$, and
 the 1st Chern class of
the line bundle $L$ is determined by the 1st Chern class of the bundle $E$.

Complex
line bundles are classified by the 1st Chern class, i.e. are isomorphic if and only if their 1rst
Chern classes coincide. Moreover, any second cohomology class is the 1rst Chern class of some
complex line bundle.

Now one checks that $L^*\bigoplus E=L^*\bigoplus L\bigoplus \C^{m-1}$, where $L^*$ is the dual bundle to $L$, has trivial 1st Chern class, and so is trivial.
\end{proof}

\begin{remark}\label{R:4D}
\indent
\begin{enumerate}
\item Under conditions of Theorem \ref{T:3Dtriv}, one sees from its proof that any $m$-bundle has a trivial $m-1$ subbundle.
\item One can show that in dimension 4, the number $m+1$ should be replaced by $m+2$.
\end{enumerate}
\end{remark}

We can now establish the main result:

\begin{theorem}\label{T:main}
Let $L$ be a self-adjoint elliptic $\G$-periodic operator in $\RR^n, n\leq 3$ and $S\subset\RR$ be the union of $m$ spectral bands of $L$, separated by gaps from the rest of the spectrum. Then there exist $m+1$  exponentially decaying composite Wannier functions
$w_j(x)$, such that their shifts $w_{j,\g}:=w_j(x-\g),\g\in\G$ form a tight (Parseval) frame in the spectral subspace $H_S$ of the operator $L$. This means that for any $f(x)\in H_S$, the equality holds
\begin{equation}
    \int\limits_{\RR^n}|f(x)|^2dx=\sum\limits_{j,\g}|\int\limits_{\RR^n}f(x)\overline{w_{j,\g}(x)}dx|^2.
\end{equation}
 A choice of $m$ such functions is possible if and only if $\Lambda_S$ is trivial, in which case an orthonormal basis of exponentially decaying composite Wannier functions exists.
\end{theorem}
We provide here just an outline of the proof, with missing details that can be found in \cite{Kuch_wan}
\begin{proof}
Consider the $m$-dimensional Bloch bundle $\Lambda_S$, which is analytic on a neighborhood $\Omega_\alpha$ (see (\ref{E:omega-a})) of the torus $\T^*\subset\C^d$, with $d\leq 3$. As it is noted in \cite{Kuch_wan}, this neighborhood is a Stein domain, contractible to the torus, and thus analytic and topological classifications of vector bundles here coincide, due to the Grauert Theorem \cite{Grauert1,Grauert2,Grauert3}. According to Theorem \ref{T:3Dtriv}, there is an analytic line bundle $L$ such that $\Lambda_S\bigoplus L$ is trivial.

Consider the infinite dimensional bundle
\begin{equation}
\Sigma_S:=\bigcup_{\Omega_a}\emph{Ran}\,(I-P(k)),
\end{equation}
complementary to $\Lambda_S$.

Due to Kuiper's theorem \cite{Kuiper}, $\Sigma_S$ is topologically, and due to Bungart's version \cite{Bungart} of the Grauert Theorem, analytically trivial as well.  As shown in \cite{ZK} (see the details in \cite{Kuch_wan}), an isomorphic copy of $L$ (which we denote with the same letter $L$) can be found in  $\Sigma_S$.
Then the $(m+1)$-subbundle $\Lambda_S\bigoplus L$ in the trivial bundle $\Omega_\alpha\times L_2(W)$ is trivializable. Thus, the construction of \cite{Kuch_wan} (picking a holomorphic basis in $\Lambda_S\bigoplus L$ and then projecting it into $\Lambda_S$) provides the required Parseval frame and thus finishes the proof.
\end{proof}

\section{Dimension four}\label{S:d4}

The four-dimensional situation is more complex, but less interesting, so we formulate the corresponding topological result here without the proof.

The answer to the question of trivialization of vector bundles over a $4$-complex is given in the following proposition.
\begin{theorem}\label{T:4Complex}
Let $E$ be a rank $m$ complex vector bundle $E$ over a $4$-complex $X$.
\begin{enumerate}
\item $E$ is trivial if and only if the total Chern class $c(E)$ is equal to $1$.
\item $E$ embeds in a trivial rank $m+1$ bundle if and only if $c_2(E) = c_1(E)^2$.
\item Any rank $m$ complex vector bundle over a $4$-complex embeds into a trivial rank $m+2$ bundle.
\end{enumerate}
\end{theorem}

Exactly as in the previous section, using Theorem \ref{T:4Complex} and following the scheme of the proof in \cite{Kuch_wan}, one obtains
\begin{theorem}\label{T:main4}
Let $L$ be a self-adjoint elliptic $\G$-periodic operator in $\RR^4$ and $S\subset\RR$ be the union of $m$ spectral bands of $L$, separated by gaps from the rest of the spectrum. Then there exist $m+2$  exponentially decaying composite Wannier functions
$w_j(x)$, such that their shifts $w_{j,\g}:=w_j(x-\g),\g\in\G$ form a tight (Parseval) frame in the spectral subspace $H_S$ of the operator $L$. This means that for any $f(x)\in H_S$, the equality holds
\begin{equation}
    \int\limits_{\RR^4}|f(x)|^2dx=\sum\limits_{j,\g}|\int\limits_{\RR^4}f(x)\overline{w_{j,\g}(x)}dx|^2.
\end{equation}
\end{theorem}

\section{Generic choice of a ``missing'' function}\label{S:generic}

Now the question arises of how one can find these $m+1$ Wannier functions. The question is non-trivial even in the absence of the topological obstacle. There are recent strong works approaching this issue (see, e.g., \cite{Cances,Cornean,CornMonaco,Fior,Monaco_etal,Read}).

Without answering this question, we make a comment here that might turn out to be useful.

Let us assume that we know how to find the Wannier functions in absence of the topological obstruction (i.e., how to find a holomorphic basis of a trivial Bloch bundle $\Lambda_S$). Then, if the obstacle is present, the proofs of Theorems \ref{T:kucwan} and \ref{T:main} show that one finds a linear bundle $L$ such that $\Lambda_S\bigoplus L$ is trivial, embeds it into $\Sigma_S$, and uses a holomorphic basis of $\Lambda_S\bigoplus L$ to produce the desired frame of Wannier functions. The question is: assuming that such $L$ is found, how can one practically imbed it into $\Sigma_S$? The idea is that a transversality theorem (e.g., \cite{KalZaid}) implies that essentially ``almost any'' bundle morphism from $L$ into $\Sigma_S$ will do.

To simplify the situation, consider a linear analytic bundle $M$, such that the 2-bundle $N:=L\bigoplus M$ is trivial. Then, an imbedding of $N$ into $\Sigma_S$ would generate an imbedding of $L$.

It is easy to find a non-trivial morphism of $M$ into $\Sigma_S$. Indeed, consider for instance a two-dimensional linear subspace $Q$ in $Ran\,(I-P(k))$ for some $k=k_0$ and consider the trivialization map from $M$ to $\Omega_\alpha\times Q$, which is an isomorphism by construction. However,  $\Omega_\alpha\times Q$ does not sit inside of $\Sigma_S$, except of the fiber at the point $k_0$. Multiplying by the projector-function $(I-P(k))$, we fix this difficulty and obtain a non-trivial bundle morphism $\Phi: M\mapsto \Sigma_S$, which might fail to be imbedding, except in a neighborhood of the base point $k_0$.

\begin{theorem}\label{T:generic}
A generic perturbation of any morphism $\Phi:M\mapsto \Sigma_S$ produces an imbedding of $M$ into $\Sigma_S$.
\end{theorem}
In fact, the statement holds also if one replaces $\Sigma_S$ (which is infinite dimensional and trivial \cite{Kuch_wan,ZK}) by any holomorphic vector bundle $R$ of a sufficiently large rank $r$. Indeed, in this case one deals with a $r\times 2$-matrix function on a three-dimensional manifold. One needs to make sure that the rank of all matrices is maximal. For a sufficiently high $r$, the sub-variety of the matrices of smaller rank has co-dimension larger than the dimension of the set of parameters $\Omega_\alpha$. Then application of a transversality theorem \cite{KalZaid} finishes the proof.

So, a ``random'' perturbation of the non-trivial morphism described above should produce an imbedding and thus, as it is described above, a Parseval frame of exponentially decaying Wannier functions. 

\section{Conclusion and remarks}\label{S:remarks}
\begin{enumerate}
\item We have shown that in the spatial dimension up to three, for any isolated $m$-band $S$ of the spectrum of an elliptic self-adjoint operator periodic with respect to a lattice $\G$, there exists a family of $m+1$ exponentially decaying composite (generalized) Wannier functions such that their $\G$-shifts form a Parseval frame in the spectral subspace $H_S$ that corresponds to $S$. This result was announced in \cite{KuchBAMS}.

    In dimension $4$, one needs $m+2$ such functions.

\item Theorem \ref{T:Planch} and Lemma \ref{L:smooth_wannier} show that the best rate of exponential decay one can expect from Wannier functions is dictated by the value $\alpha$ in the domain $\Omega_\alpha$ of analyticity of the spectral projectors $P(k)$ (the larger $\alpha$, the faster is the decay). The proof of the main result of this paper shows that the Parseval frames allow for the same best rate of decay.

    One can improve on the decay rate by allowing anisotropic exponential decay, and thus using a multi-index $\alpha=(\alpha_1,\dots,\alpha_n)$ and corresponding domains $\Omega_\alpha$ (see Remark \ref{R:aniso}).

\item The proof of Theorem \ref{T:3Dtriv} shows (see Remark \ref{R:4D}) that when the spatial dimension does not exceed $3$, $m-1$ holomorphic linearly independent sections can be found in $\Lambda_S$, even if the bundle is non-trivial. So, it is just one family of Wannier functions that is missing and has to be replaced by two.

  \item The original restriction made on the operator is not essential. The results also apply in an abstract situation of a periodic elliptic self-adjoint operator on an abelian covering of a compact manifold, graph, or quantum graph, as long as the torsion free rank of the abelian deck group (rather than the spatial dimension) does not exceed 4. Neither formulations, nor proofs require any modifications in this case.

  \item The results of this article are of pure existence nature and do not provide any answer to the question of how one can actually construct the required number of the Wannier functions. However, a remark on generic nature of the required functions is made.
\end{enumerate}

\section*{Acknowledgments}
The second author expresses his gratitude to the NSF grants DMS-0406022 and DMS-1517938 for partial support of this work and to J.~Corbin, A.~Levitt, D. Monaco, G.~Panati, and M.~Zaidenberg for useful discussions.

\end{document}